\documentclass[runningheads]{llncs}
\usepackage{xspace}
\usepackage{array}
\usepackage{xcolor}
\usepackage{epsfig}
\usepackage{gensymb}
\usepackage[english]{babel}
\usepackage{graphicx}
\usepackage{amssymb}
\usepackage{multirow}
\usepackage{todonotes}
\usepackage{enumerate}
\usepackage{paralist}
\usepackage{amsmath}
\usepackage{siunitx}
\usepackage{graphicx}
\usepackage{xspace}
\usepackage{amsmath}
\usepackage{color}
\usepackage{amssymb}
\usepackage{amsmath}
\usepackage{paralist}
\usepackage[colorlinks=true]{hyperref}
\hypersetup{colorlinks,linkcolor={blue},citecolor={blue},urlcolor={blue}} 
\newcommand{\bigoh}{\ensuremath{\mathcal{O}}}

\usepackage[caption=false]{subfig} 

\usepackage{lineno}

\let\doendproof\endproof
\renewcommand\endproof{~\hfill\qed \doendproof}

\newcommand{\remove}[1]{}

\newcommand{\partition}{{\sc Partition}\xspace}
\newcommand{\upward}{{\sc Upward Planarity Testing}\xspace}

\usepackage{fullpage}

\begin{document}
	%
	%
	%
	\title{Upward Planarity Testing of \\Biconnected Outerplanar DAGs \\Solves Partition\thanks{Research partially supported by PRIN projects no.\ 2022ME9Z78 ``NextGRAAL: Next-generation algorithms for constrained GRAph visuALization'' and no.\  2022TS4Y3N ``EXPAND: scalable algorithms for EXPloratory Analyses of heterogeneous and dynamic Networked Data''.}}
	\author{ Fabrizio Frati
	}
	%
	\titlerunning{Upward Planarity of Outerplanar DAGs and Partition}
	%
	\institute{Roma Tre University, Rome, Italy, 	\email{fabrizio.frati@uniroma3.it}}
	
	\maketitle

\begin{abstract}
We show an $\bigoh(n)$-time reduction from the problem of testing whether a multiset of positive integers can be partitioned into two multisets so that the sum of the integers in each multiset is equal to $n/2$ to the problem of testing whether an $n$-vertex biconnected outerplanar DAG admits an upward planar drawing. This constitutes the first barrier to the existence of efficient algorithms for testing the upward planarity of DAGs with no large triconnected minor.    

We also show a result in the opposite direction. Suppose that partitioning a multiset of positive integers into two multisets so that the sum of the integers in each multiset is $n/2$ can be solved in $f(n)$ time. Let $G$ be an $n$-vertex biconnected outerplanar DAG and $e$ be an edge incident to the outer face of an outerplanar drawing of $G$. Then it can be tested in $\bigoh(f(n))$ time whether $G$ admits an upward planar drawing with $e$ on the outer face.
\end{abstract}

\section{Introduction}\label{le:intro}

An \emph{upward planar drawing} of a directed acyclic graph (DAG, for short) maps each vertex to a point of the plane and each edge to a Jordan arc monotonically increasing from the source to the sink of the edge, so that no two edges intersect, except at common end-vertices. Upward planarity is perhaps the most natural extension of planarity to directed graphs. As such, the problem of testing whether a DAG $G$ admits an upward planar drawing has received an enormous attention in the literature. While the problem was early shown to be NP-hard~\cite{gt-ccurpt-01}, it is known to be tractable in many important cases, among which: (i) if $G$ has a fixed combinatorial embedding, that the required upward planar drawing has to respect~\cite{bdl-udtd-94}; (ii) if $G$ has a single source~\cite{DBLP:journals/siamcomp/BertolazziBMT98,BrucknerHR19,HuttonL91,HuttonL96} or, more in general, a bounded number of sources~\cite{cdf-up-22}; (iii) if $G$ contains no large triconnected minor~\cite{DBLP:conf/gd/ChaplickGFGRS22,DBLP:journals/siamdm/DidimoGL09,Papakostas94}. The parameterized complexity of the problem has also been investigated~\cite{Chan04,cdf-up-22,DBLP:journals/siamdm/DidimoGL09,HealyL06}.

This paper focuses on \emph{outerplanar DAGs} -- fitting case~(iii) in the above list. These are DAGs whose underlying (undirected) graph admits an \emph{outerplanar drawing}, i.e., a planar drawing in which all the vertices are incident to the outer face. An $\bigoh(n^2)$-time algorithm for testing whether an $n$-vertex outerplanar DAG admits a (not necessarily outerplanar) upward planar drawing is known~\cite{Papakostas94}. Although we do not fully understand the details of the mentioned algorithm~\cite{Papakostas94}, an $\bigoh(n^2)$-time upward planarity testing algorithm is now known~\cite{DBLP:conf/gd/ChaplickGFGRS22} for all the $n$-vertex DAGs with no $K_4$-minor -- this is the class of \emph{directed partial $2$-trees}, which includes the one of the outerplanar DAGs. 

Can we design an upward planarity test for outerplanar DAGs running in linear time? We do not know. However, in this paper, we link the complexity of the problem to the one of a very famous problem in complexity theory. The \partition problem takes as input a multiset $\mathcal S$ of positive integers and asks whether $\mathcal S$ admits a partition into two multisets such that the sum of the elements in each multiset is the same\footnote{In the rest of the paper, we use the term ``set'' in place of ``multiset'', implicitly assuming that sets might contain multiple occurrences of the same element.}. \partition is one of the ``six basic NP-complete problems'' in the notorious book~\cite{DBLP:books/fm/GareyJ79} by Garey and Johnnson. In the same book, it is proved that the problem admits a pseudo-polynomial-time algorithm, that is, an algorithm which runs in polynomial time if each element, or equivalently the sum of the elements, is bounded by a polynomial function of the number of elements. Specifically, there is an algorithm based on dynamic programming that solves \partition in $\bigoh(k\cdot n)$ time, if we denote by $k$ the number of elements and by $n$ the sum of the elements in the set\footnote{In the literature, the number of elements is denoted by $n$ and their sum by $B$, however as will be evident soon, our notation is more convenient for this paper.}. As $k$ might be in $\Omega(n)$, the running time of the algorithm is $\bigoh(n^2)$, when expressed solely as a function of $n$. The time needed to solve \partition is nowadays known to be $\bigoh(n \log^2 n)$, due to a recent algorithm for the more general {\sc{Subset Sum}} problem~\cite{DBLP:journals/talg/KoiliarisX19}. 

Our main result is the following. 

\begin{theorem}\label{th:partition-to-upward}
	Let $\mathcal S=\{a_1,\dots,a_k\}$ be an instance of the \partition problem, where $\sum_{i=1}^k a_i=n$. It is possible to construct in $\bigoh(n)$ time an $\bigoh(n)$-vertex biconnected outerplanar DAG $G$ which is upward planar if and only if $\mathcal S$ is a positive instance of  \partition. Moreover, a partition of $\mathcal S$ into two sets $\mathcal S_1$ and $\mathcal S_2$ such that  $\sum_{a_i\in \mathcal S_1} a_i=\sum_{a_i\in \mathcal S_2} a_i=n/2$ can be constructed in $\bigoh(n)$ time from an upward plane embedding of  $G$.
\end{theorem}

Theorem~\ref{th:partition-to-upward} provides evidence that an $\bigoh(n)$-time algorithm for testing the upward planarity of an $n$-vertex biconnected outerplanar DAG might be very hard to obtain, as it would imply an $\bigoh(n)$-time algorithm for \partition, which would be a ground-breaking result. We also remark that, due to Theorem~\ref{th:partition-to-upward}, a super-linear lower bound for \partition would imply the same lower bound for testing the upward planarity of a biconnected outerplanar DAG. The authors of this paper are unaware of any graph drawing problem that does not involve real numbers (e.g., a prescribed point set), that is not NP-hard, and whose complexity is super-linear. Hence, this would constitute an important step towards establishing polynomial lower bounds for the complexity of graph drawing problems.

We can also prove a reduction in the opposite direction, as in the following.

\begin{theorem}\label{th:upward-to-partition}
Suppose that an algorithm exists that solves an instance $\mathcal S$ of the \partition problem in $f(n)$ time, for some function $f(n)\in \Omega(n)$, where $n$ denotes the sum of the elements in $\mathcal S$. Then, for any $n$-vertex biconnected outerplanar DAG $G$ and for any edge $e$ incident to the outer face of an outerplanar drawing of $G$, it can be tested in $\bigoh(f(n))$ time whether $G$ admits an upward planar drawing in which $e$ is incident to the outer face.
\end{theorem}

We make three remarks on Theorem~\ref{th:upward-to-partition}. First, together with the mentioned $\bigoh(n \log^2 n)$-time algorithm to solve the \partition problem~\cite{DBLP:journals/talg/KoiliarisX19}, Theorem~\ref{th:upward-to-partition} implies the following. 

\begin{theorem}\label{th:solving-outerplanar}
For any $n$-vertex biconnected outerplanar DAG $G$ and for any edge $e$ incident to the outer face of an outerplanar drawing of $G$, there exists an algorithm that tests in $\bigoh(n \log^2 n)$ time whether $G$ admits an upward planar drawing in which $e$ is incident to the outer face.
\end{theorem}

Second, although its statement does not mention this explicitly, the proof of Theorem~\ref{th:partition-to-upward} shows that, given an instance $\mathcal S=\{a_1,\dots,a_k\}$ of \partition with $\sum_{i=1}^k a_i=n$, it is possible to construct in $\bigoh(n)$ time an $\bigoh(n)$-vertex biconnected outerplanar DAG $G$ such that, {\em for a certain edge $e$ incident to the outer face of an outerplanar drawing of $G$}, we have that $G$ admits an upward planar drawing {\em in which $e$ is incident to the outer face} if and only if $\mathcal S$ is a positive instance of  \partition. This, together with Theorem~\ref{th:upward-to-partition}, implies that the \partition problem is linear-time equivalent to the problem of testing whether a biconnected outerplanar DAG $G$ with a given edge $e$ incident to the outer face of an outerplanar drawing of $G$ admits an upward planar drawing in which $e$ is incident to the outer face. 

Third, the approach of solving an algorithmic graph drawing problem by first constraining a prescribed edge to be incident to the outer face, and by only later getting rid of that constraint, nicely interfaces with the decomposition of a biconnected planar graph into its triconnected components. This decomposition is usually expressed by a tree, namely the SPQR-tree for a biconnected planar graph, or the SPQ-tree for a biconnected partial $2$-tree, or the (extended) dual tree for a biconnected outerplanar graph; having a prescribed edge incident to the outer face corresponds to rooting the decomposition tree at some node, which then allows for a bottom-up computation on the tree. This approach has been successfully used for problems related to ours, namely testing the upward planarity of directed partial $2$-trees~\cite{DBLP:conf/gd/ChaplickGFGRS22} and testing the rectilinear planarity of outerplanar graphs~\cite{DBLP:journals/comgeo/Frati22}; for these two problems, removing the constraint about the prescribed edge on the outer face is doable without increasing the algorithm's running time, by using a strategy devised in~\cite{dlop-ood-20}.

{\bf Related results.} Fast approximation and exact algorithms for \partition are known~\cite{kk-dm-82,DBLP:books/daglib/0010031,DBLP:conf/ijcai/Korf11}, running in time exponential in $k$~\cite{DBLP:journals/jacm/HorowitzS74,k-caa-98,DBLP:conf/aips/KorfS13,DBLP:journals/siamcomp/SchroeppelS81}. See the PhD thesis by Schreiber~\cite{DBLP:phd/basesearch/Schreiber14} for more references. 

\partition is a special case of the {\sc {Subset Sum}} problem, whose input consists of a set $\mathcal S$ of $k$ positive integers that sum up to $n$, and of a target integer $t$. The problem asks whether there exists a subset of $\mathcal S$ whose elements sum up to $t$; note that, when $t=n/2$, this is exactly the \partition problem. A plethora of results about the {\sc {Subset Sum}} problem are known. Here we mention only a few of them. Let $k'$ and $m$ be the number of distinct elements in $\mathcal S$ and the value of the maximum element in $\mathcal S$, respectively. Pisinger~\cite{DBLP:journals/jal/Pisinger99} presented an $\bigoh(n m)$-time algorithm for the {\sc {Subset Sum}} problem. Koiliaris and Xu~\cite{DBLP:journals/talg/KoiliarisX19} presented an algorithm running in $\bigoh\left(\min\left \{\sqrt{k'}m \log^{5/2}m, m^{4/3}\log^2 m, n \log n \log (k' \log m)\right\}\right)$ time; note that the last of the three terms in the previous minimum gives an $\bigoh(n\log^2 n)$ upper bound, as $k'\in \bigoh(n)$ and $m\in \bigoh(n)$; this is the previously mentioned best-known upper bound for the \partition problem. Bringmann~\cite{DBLP:conf/soda/Bringmann17} presented a randomized algorithm running in $\bigoh(k+t\log t \log^3 (k/\delta)\log k)$ time for computing a set that contains every value $t'\leq t$ with probability larger than $1-\delta$.

The complexity of upward planarity testing is intertwined with the one of \emph{rectilinear planarity testing}, which asks whether a graph admits a planar drawing in which each edge is horizontal or vertical. The reductions proving NP-hardness for rectilinear planarity testing and upward planarity testing are similar~\cite{gt-ccurpt-01} and, like upward planarity testing, rectilinear planarity testing is tractable when the graph has a fixed combinatorial embedding~\cite{bkmnw-msms-17,t-eggmnb-87} and when it contains no large triconnected minor~\cite{blr-odie-16,dlv-sood-98,glm-sr-19}. Our results highlight a possible divergence between upward and rectilinear planarity testing, as the rectilinear planarity of an outerplanar graph can be tested in linear time~\cite{DBLP:journals/comgeo/Frati22}. 

\section{Preliminaries} \label{se:preliminaries}

A \emph{plane embedding} of a biconnected planar graph $G$ is an equivalence class of planar drawings of $G$, where two drawings $\Gamma_1$ and $\Gamma_2$ are equivalent if the clockwise order of the edges incident to each vertex is the same in $\Gamma_1$ and $\Gamma_2$, and the clockwise order of the vertices along the cycle delimiting the outer face is the same in $\Gamma_1$ and $\Gamma_2$. We often talk about \emph{faces of a plane embedding}, meaning faces of any planar drawing in that equivalence class. If the drawings in the class are outerplanar, we have an \emph{outerplane embedding}; this is actually unique, up to an inversion of all the orders defining it, for a biconnected outerplanar graph. We denote by $f_{\mathcal E}$ the outer face of a plane embedding $\mathcal E$.


Let $G$ be a DAG. A \emph{switch} of $G$ is a vertex that is a source or a sink. An upward planar drawing of $G$ determines an assignment of labels to the angles of the corresponding plane embedding, where an \emph{angle} $\alpha$ at a vertex $u$ in a face $f$ of a plane embedding represents an incidence of $u$ on $f$. In particular, $\alpha$ is \emph{flat} and has label $0$ if the edges delimiting it are one incoming into $u$ and one outgoing from $u$. Otherwise, $\alpha$ is a \emph{switch} angle and it is \emph{small} (with label $-1$) or \emph{large} (with label $1$) depending on whether the (geometric) angle at $f$ representing $\alpha$ is smaller or larger than $180^\circ$, respectively. An \emph{upward plane embedding} of $G$ is an equivalence class of upward planar drawings of $G$, where two drawings are equivalent if they determine  the same plane embedding $\mathcal E$ for $G$ and  the same label assignment for the angles of $\mathcal E$. The label assignments that enhance a plane embedding to an upward plane embedding are characterized as follows. 

\begin{theorem}[\cite{bdl-udtd-94,DBLP:journals/siamdm/DidimoGL09}]\label{th:upward-conditions}
	Let $G$ be a digraph with plane embedding  $\mathcal E$, and $\lambda$ be a label assignment for the angles of $\mathcal E$. Then $\mathcal E$ and $\lambda$ define an upward plane embedding of $G$ if and only if the following hold:
	\begin{enumerate}[]
		\item{\bf (UP0)} If $\alpha$ is a switch angle then $\alpha$ is small or large, otherwise it is flat.
		\item{\bf (UP1)} If $v$ is a switch vertex, the number of small, flat and large angles incident to $v$ is equal to $\deg(v)-1$, $0$, and $1$, respectively.
		\item{\bf (UP2)} If $v$ is a non-switch vertex, the number of small, flat and large angles incident to $v$ is equal to $\deg(v)-2$, $2$, and $0$, respectively.
		\item{\bf (UP3)} If $f$ is an internal face (the outer face) of $\mathcal E$, the number of small angles in $f$ is equal to the number of large angles in $f$ plus $2$ (resp.\ minus $2$).
	\end{enumerate}
\end{theorem}	

Let $(\mathcal E,\lambda)$ be an upward plane embedding and let $u$ be a vertex incident to $f_{\mathcal E}$. We often talk about \emph{the angle at $u$ inside $\mathcal E$}, referring to the complement of the angle at $u$ in $f_{\mathcal E}$. Then the angle at $u$ inside $\mathcal E$ is small, flat, or large if and only if the angle at $u$ in $f_{\mathcal E}$ is large, flat, or small, respectively.

The \upward problem asks whether a given DAG admits an upward planar drawing. Because of Theorem~\ref{th:upward-conditions}, this is equivalent to deciding whether the DAG admits a plane embedding $\mathcal E$ and a label assignment for the angles of $\mathcal E$ satisfying Conditions (UP0)--(UP3). 

\section{Solving Partition Via Upward Planarity} \label{se:from-partition-to-upward}

In this section we prove Theorem~\ref{th:partition-to-upward}. 

Let $\mathcal S=\{a_1,\dots,a_k\}$ be an instance of the \partition problem, where $\sum_{i=1}^k a_i=n$. We start by describing the construction of an $\bigoh(n)$-vertex biconnected outerplanar DAG $G$ from $\mathcal S$. Refer to Fig.~\ref{fig:partition-to-upward}. The structure of $G$ is simple, indeed it consists of a cycle $\mathcal C$ and of some other cycles, each sharing an edge with $\mathcal C$ and disjoint with each other. That is, the weak dual of an outerplane embedding $\mathcal O$ of $G$ is a star, where the \emph{weak dual} of $\mathcal O$ is the graph that has a vertex for each internal face of $\mathcal O$ and an edge between two vertices if the corresponding faces share an edge. 

\begin{figure}[htb]
	\centering
	\includegraphics[page=1, width=0.9\textwidth]{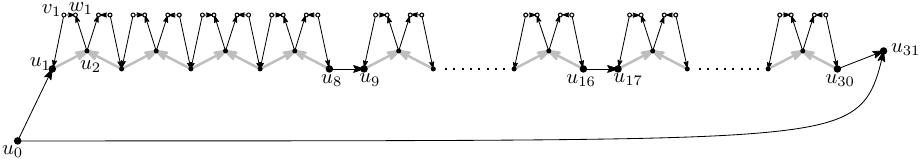}
	\caption{Construction of $G$ from $\mathcal S=\{1,1,2\}$; then $n=4$ and $k=3$. Vertices of $\mathcal C$ are filled disks, while other vertices are empty disks. Larger disks represent $u_0$, $u_{6n+2k+1}$, and the initial and final vertices of the paths $\mathcal P_1,\dots,\mathcal P_k$, which are represented by thicker gray lines.}
	\label{fig:partition-to-upward}
\end{figure}

We initialize $G$ to a cycle $\mathcal C=(u_0,u_1,\dots,u_{6n+2k+1})$. Let $\mathcal P$ be the path $(u_1,u_2,\dots,u_{6n+2k})$ in $\mathcal C$. Let $\mathcal P_1$ be the path that comprises the first $6a_1+2$ vertices $(u_1,\dots,u_{6a_1+2})$ of $\mathcal P$, let $\mathcal P_2$ be the path that comprises the next $6a_2+2$ vertices $(u_{6a_1+3},\dots,u_{6a_1+6a_2+4})$ of $\mathcal P$, and so on till $\mathcal P_k$. For each edge $(u_j,u_{j+1})$ of $\mathcal C$ that belongs to a path $\mathcal P_i$, we introduce in $G$ two vertices $v_j$ and $w_j$ and the edges of a path $\mathcal Q_j:=(u_j,v_j,w_j,u_{j+1})$; we denote by $\mathcal C_j$ the cycle $\mathcal Q_j \cup (u_j,u_{j+1})$. Finally, we direct the edges of $G$ as follows.  
\begin{itemize}
	\item First, for $i=1,\dots,k$, each edge $(u_j,u_{j+1})$ of $\mathcal P_i$ is outgoing from the vertex with odd index. 
	\item Second, for $i=1,\dots,k-1$, the edge connecting the last vertex of $\mathcal P_i$ to the first vertex of $\mathcal P_{i+1}$ is outgoing from the former vertex. 
	\item Third, both the edges incident to $u_0$ are outgoing from $u_0$ and both the edges incident to $u_{6n+2k+1}$ are incoming into $u_{6n+2k+1}$. 
	\item Finally, the edges of the path $\mathcal Q_j$ are outgoing from $v_j$ and incoming into $w_j$ if $j$ is odd, while they are incoming into $v_j$ and outgoing from $w_j$ if $j$ is even.
\end{itemize} 

Clearly, $G$ is an $\bigoh(n)$-vertex biconnected outerplanar DAG and can be constructed in $\bigoh(n)$ time. Next, we prove that $\mathcal S$ is a positive instance of \partition if and only if $G$ is a positive instance of the \upward problem. For $i=1,\dots,k$, let $\mathcal V_i$ be the set that comprises the vertices of $\mathcal P_i$ and the vertices of the paths $\mathcal Q_j$ such that $(u_j,u_{j+1})$ belongs to $\mathcal P_i$.

$(\Longrightarrow)$ Suppose first that there exists a partition  of $\mathcal S$ into two sets $\mathcal S_1$ and $\mathcal S_2$ such that  $\sum_{a_i\in \mathcal S_1} a_i=\sum_{a_i\in \mathcal S_2} a_i=n/2$. We construct a plane embedding $\mathcal E$ of $G$ as follows. Starting from any plane embedding of $\mathcal C$, we embed inside $\mathcal C$ any path $\mathcal Q_j$ such that the edge $(u_j,u_{j+1})$ belongs to a path $\mathcal P_i$ with $a_i \in \mathcal S_1$; also, we embed outside $\mathcal C$ any path $\mathcal Q_j$ such that the edge $(u_j,u_{j+1})$ belongs to a path $\mathcal P_i$ with $a_i \in \mathcal S_2$, so that $\mathcal C_j$ delimits an internal face of $\mathcal E$. Next, we specify a label assignment $\lambda$ for the angles of $\mathcal E$. 
\begin{figure}[htb]\tabcolsep=4pt
	\centering
	\begin{tabular}{c c c}
		\includegraphics[page=4, scale=.9]{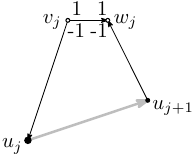} \hspace{3mm} &
		\includegraphics[page=5, scale=.9]{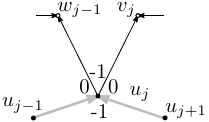} \hspace{3mm} &
		\includegraphics[page=6, scale=.9]{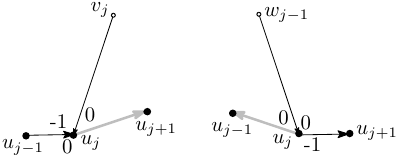} \\
		(a) \hspace{3mm} & (b) \hspace{3mm} & (c)\\
	\end{tabular}
	\caption{Labels for some angles of $\mathcal E$. (a) Switch vertices $v_j$ and $w_j$. (b) Vertices $u_j$ that are internal to some path $\mathcal P_i$. (c) Vertices $u_j$ that are the first or last vertex of some path $\mathcal P_i$.}
	\label{fig:labels}
\end{figure}

\begin{itemize}
	\item First, consider any degree-$2$ switch vertex $v_j$ (resp.\ $w_j$); see Fig.~\ref{fig:labels}(a). We label $-1$ the angle at $v_j$ (resp.\ at $w_j$) in the internal face of $\mathcal E$ delimited by $\mathcal C_j$, and $+1$ the other angle $v_j$ (resp.\ $w_j$) is incident to; this assignment respects Conditions (UP0), (UP1), and, vacuously, (UP2) of Theorem~\ref{th:upward-conditions}. 
	\item Second, consider any degree-$4$ non-switch vertex $u_j$ that is an internal vertex of a path $\mathcal P_i$; see Fig.~\ref{fig:labels}(b). We  label $0$ the angles at $u_j$ in the internal faces of $\mathcal E$ delimited by $\mathcal C_j$ and $\mathcal C_{j+1}$, and $-1$ the other two angles $u_j$ is incident to; this assignment respects Conditions (UP0), (UP2), and, vacuously, (UP1)  of Theorem~\ref{th:upward-conditions}.
	\item Third, consider any degree-$3$ non-switch vertex $u_j$ that is the first (resp.\ last) vertex  of a path $\mathcal P_i$; see Fig.~\ref{fig:labels}(c). We label $0$ the angle at $u_j$ in the internal face of $\mathcal E$ delimited by $\mathcal C_j$ (resp.\ by $\mathcal C_{j-1}$). If $\mathcal C_j$ (resp.\ $\mathcal C_{j-1}$) is embedded inside $\mathcal C$, then we label $0$ the angle at $u_j$ in the outer face of $\mathcal E$ and $-1$ the remaining angle $u_j$ is incident to, otherwise we label $-1$ the angle at $u_j$ in the outer face of $\mathcal E$ and $0$ the remaining angle $u_j$ is incident to; this assignment respects Conditions (UP0), (UP2), and, vacuously, (UP1)  of Theorem~\ref{th:upward-conditions}.
	\item Finally, consider the degree-$2$ switch vertex $u_0$ (resp.\ $u_{6n+2k+1}$). We label $+1$ the angle at $u_0$ (resp.\ $u_{6n+2k+1}$) in the outer face of $\mathcal E$, and $-1$ the other angle $u_0$ (resp.\ $u_{6n+2k+1}$) is incident to; this assignment respects Conditions (UP0), (UP1), and, vacuously, (UP2) of Theorem~\ref{th:upward-conditions}. 
\end{itemize}  

In order to prove that $(\mathcal E,\lambda)$ is an upward plane embedding of $G$, it remains to prove that $\lambda$ satisfies Condition (UP3) of Theorem~\ref{th:upward-conditions}. Consider any internal face delimited by a cycle $\mathcal C_j$. Two angles in such a face, those at $v_j$ and $w_j$, are labeled $-1$, while the other two angles, those at $u_j$ and $u_{j+1}$, are labeled $0$, hence $\lambda$ respects Condition (UP3) for the considered face. Two faces remain to be considered, namely the outer face $f_{\mathcal E}$ and the internal face $g_{\mathcal E}$ incident to $u_0$. In order to deal with these faces, we need to observe the following facts. 
\begin{itemize}
	\item If $a_i$ belongs to $\mathcal S_1$, the number of angles of $f_{\mathcal E}$ labeled $-1$ at the vertices in $\mathcal V_i$ is $6a_i$, while the number of angles of $f_{\mathcal E}$ labeled~$+1$ at the vertices in $\mathcal V_i$ is~$0$. Indeed, all the angles of $f_{\mathcal E}$ at the vertices of $\mathcal P_i$ are labeled~$-1$, except for the angles at the end-vertices of $\mathcal P_i$, which are labeled~$0$. The vertices $v_j$ and $w_j$ in $\mathcal V_i$ are not incident to $f_{\mathcal E}$. Analogously, if $a_i$ belongs to $\mathcal S_2$, the number of angles of $g_{\mathcal E}$ labeled~$-1$ at the vertices in $\mathcal V_i$ is $6a_i$, while the number of angles of $g_{\mathcal E}$ labeled~$+1$ at the vertices in $\mathcal V_i$ is~$0$. 
	\item If $a_i$ belongs to $\mathcal S_1$, the number of angles of $g_{\mathcal E}$ labeled~$-1$ at the vertices in $\mathcal V_i$ is $6a_i+2$, while the number of angles of $g_{\mathcal E}$ labeled~$+1$ at the vertices in $\mathcal V_i$ is $12a_i+2$. Indeed, all the angles of $g_{\mathcal E}$ at the vertices of $\mathcal P_i$ are labeled~$-1$ and all the angles of $g_{\mathcal E}$ at the vertices $v_j$ and $w_j$ in $\mathcal V_i$ are labeled~$+1$ (note that there are $6a_i+1$ cycles $\mathcal C_j$ such that the edge $(u_j,u_{j+1})$ belongs to $\mathcal P_i$). Analogously, if $a_i$ belongs to $\mathcal S_2$, then the number of angles of $f_{\mathcal E}$ labeled~$-1$ at the vertices in $\mathcal V_i$ is $6a_i+2$, while the number of angles of $f_{\mathcal E}$ labeled~$+1$ at the vertices in $\mathcal V_i$ is $12a_i+2$. 
\end{itemize}  

It follows that the number of angles of $f_{\mathcal E}$ labeled $+1$ is $2+\sum_{a_i\in \mathcal S_2} (12a_i+2)$, where the first term of the sum comes from the angles at $u_0$ and $u_{6n+2k+1}$, while the number of angles of $f_{\mathcal E}$ labeled $-1$ is $\sum_{a_i\in \mathcal S_1} (6a_i) + \sum_{a_i\in \mathcal S_2} (6a_i+2)$. Thus, the number of angles of $f_{\mathcal E}$ labeled $+1$ minus the number of angles of $f_{\mathcal E}$ labeled $-1$ is $2+\sum_{a_i\in \mathcal S_2} (6a_i)-\sum_{a_i\in \mathcal S_1} (6a_i)=2$, where the last equality follows from $\sum_{a_i\in \mathcal S_1} a_i=\sum_{a_i\in \mathcal S_2} a_i$. Hence, $\lambda$ respects Condition (UP3) of Theorem~\ref{th:upward-conditions} for $f_{\mathcal E}$. Analogously, the number of angles of $g_{\mathcal E}$ labeled $+1$ is $\sum_{a_i\in \mathcal S_1} (12a_i+2)$, while the number of angles of $g_{\mathcal E}$ labeled $-1$ is $2+\sum_{a_i\in \mathcal S_1} (6a_i+2) + \sum_{a_i\in \mathcal S_2} (6a_i)$, where the first term of the sum comes from the angles at $u_0$ and $u_{6n+2k+1}$. Thus, the number of angles of $g_{\mathcal E}$ labeled $-1$ minus the number of angles of $g_{\mathcal E}$ labeled $+1$ is $2-\sum_{a_i\in \mathcal S_1} (6a_i)+\sum_{a_i\in \mathcal S_2} (6a_i)=2$, where the last equality follows from $\sum_{a_i\in \mathcal S_1} a_i=\sum_{a_i\in \mathcal S_2} a_i$. Hence, $\lambda$ respects Condition (UP3) of Theorem~\ref{th:upward-conditions} for $g_{\mathcal E}$.

$(\Longleftarrow)$ Suppose now that $G$ admits an upward plane embedding $(\mathcal E, \lambda)$. We construct a partition of $\mathcal S$ into two sets $\mathcal S_1$ and $\mathcal S_2$ such that  $\sum_{a_i\in \mathcal S_1} a_i=\sum_{a_i\in \mathcal S_2} a_i=n/2$ as follows. For $i=1,\dots,k$, denote by $\mathcal Q^*(i)$ the path $\mathcal Q_{j}$ that forms a cycle with the first edge $(u_j,u_{j+1})$ of $\mathcal P_i$. If $\mathcal Q^*(i)$ lies inside $\mathcal C$ in $\mathcal E$, then we put $a_i$ into $\mathcal S_1$, otherwise into $\mathcal S_2$. Clearly, the partition can be constructed in $\bigoh(n)$ time. Observe that every cycle $\mathcal C_j$ delimits an internal face of $\mathcal E$, except, possibly, for a single cycle $\mathcal C_j$ which might delimit the outer face of $\mathcal E$. If that is the case, we re-embed the path $\mathcal Q_j$ so that every cycle $\mathcal C_j$ delimits an internal face of $\mathcal E$. This modification does not alter whether any path $\mathcal Q_j$ is embedded inside or outside $\mathcal C$; see Fig.~\ref{fig:small-cycles-internal}.

\begin{figure}[htb]
	\centering
	\includegraphics[scale=1,page=7]{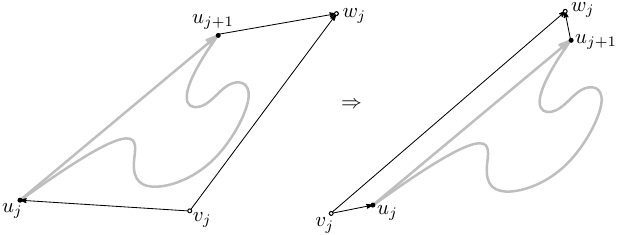}
	\caption{The embedding $\mathcal E$ can be modified by re-embedding a single path $\mathcal Q_j$ so that every cycle $\mathcal C_j$ delimits an internal face of $\mathcal E$.}
	\label{fig:small-cycles-internal}
\end{figure}

We now prove that $\sum_{a_i\in \mathcal S_1} a_i=n/2$, which implies that $\sum_{a_i\in \mathcal S_2} a_i=n/2$. As before, let $f_{\mathcal E}$ be the outer face of $\mathcal E$ and $g_{\mathcal E}$ be the internal face of $\mathcal E$ incident to $u_0$. We are going to use the following lemmata. 

\begin{figure}[htb]
	\centering
	\includegraphics[scale=1,page=2]{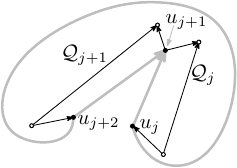}
	\caption{If $\mathcal Q_j$ lies inside $\mathcal C$ (gray in the figure) in $\mathcal E$, then $\mathcal Q_{j+1}$ also lies  inside $\mathcal C$ in $\mathcal E$.}
	\label{fig:partition-to-upward-angles}
\end{figure}

\begin{lemma}	\label{le:all-inside}
For any $i\in \{1,\dots,k\}$, either all the paths $\mathcal Q_j$ such that $(u_j,u_{j+1})$ belongs to $\mathcal P_i$ lie inside $\mathcal C$, or they all lie outside $\mathcal C$ in $\mathcal E$.
\end{lemma} 

\begin{proof}
Consider any index $j$ such that $\mathcal Q_j$ and $\mathcal Q_{j+1}$ both exist. We prove that, if $\mathcal Q_j$ lies  inside $\mathcal C$ in $\mathcal E$, then $\mathcal Q_{j+1}$ also lies  inside $\mathcal C$ in $\mathcal E$; refer to Fig.~\ref{fig:partition-to-upward-angles}. An analogous proof shows that, if $\mathcal Q_j$ lies outside $\mathcal C$ in $\mathcal E$, then $\mathcal Q_{j+1}$ also lies outside $\mathcal C$ in $\mathcal E$. Note that this claim implies the statement of the lemma. Let $(\mathcal E_{\mathcal C}, \lambda_{\mathcal C})$ be the upward plane embedding of $\mathcal C$ which is the restriction of $(\mathcal E, \lambda)$ to $\mathcal C$. Since $u_{j+1}$ is a non-switch vertex in $\mathcal C_j$ and in $\mathcal C_{j+1}$, the angles at $u_{j+1}$ in the faces delimited by $\mathcal C_j$ and by $\mathcal C_{j+1}$ are both flat in $(\mathcal E, \lambda)$, by Condition~(UP0) of Theorem~\ref{th:upward-conditions}. Since $u_{j+1}$ is a switch vertex in $\mathcal C$, the planarity of $\mathcal E$ implies that the angles at $u_{j+1}$ inside and outside $\mathcal C$ in $(\mathcal E_{\mathcal C}, \lambda_{\mathcal C})$ are large and small, respectively. Then the planarity of $\mathcal E$ implies that $\mathcal C_{j+1}$ lies inside $\mathcal C$.
\end{proof}

\begin{lemma} \label{le:contribution-inside}
For any $i\in \{1,\dots,k\}$, suppose that $\mathcal Q^*(i)$ lies inside $\mathcal C$ in $\mathcal E$. Then:
\begin{itemize}
	\item the number of angles of $f_{\mathcal E}$ labeled $-1$ at the vertices in $\mathcal V_i$ is $6a_i$;
	\item the number of angles of $f_{\mathcal E}$ labeled~$+1$ at the vertices in $\mathcal V_i$ is~$0$;
	\item the number of angles of $g_{\mathcal E}$ labeled~$-1$ at the vertices in $\mathcal V_i$ is $6a_i+2$;
	\item the number of angles of $g_{\mathcal E}$ labeled~$+1$ at the vertices in $\mathcal V_i$ is $12a_i+2$.
\end{itemize}  
\end{lemma} 

\begin{proof}
By Lemma~\ref{le:all-inside}, all the paths $\mathcal Q_j$ such that the edge $(u_j,u_{j+1})$ belongs to $\mathcal P_i$ lie inside $\mathcal C$ in $\mathcal E$. We first argue about the angles at the vertices in $\mathcal V_i$ in $(\mathcal E, \lambda)$. Every cycle $\mathcal C_j$ such that the edge $(u_j,u_{j+1})$ belongs to $\mathcal P_i$ delimits an internal face $f_j$ of $\mathcal E$. Since $u_j$ and $u_{j+1}$ are degree-$2$ non-switch vertices of $\mathcal C_j$, by Condition~(UP0) of Theorem~\ref{th:upward-conditions}, the angles at $u_j$ and $u_{j+1}$ in $f_j$ are both flat in $(\mathcal E, \lambda)$. By Condition~(UP3), the angles at $v_j$ and $w_j$ in $f_j$ are both small in $(\mathcal E, \lambda)$, and thus, by Condition~(UP1), the angles at $v_j$ and $w_j$ in $g_{\mathcal E}$ are both large in $(\mathcal E, \lambda)$. Every vertex  $u_j$ in $\mathcal V_i$ which is neither the first nor the last vertex of $\mathcal P_i$ has two incident flat angles in $f_{j-1}$ and $f_j$, hence by Condition~(UP2), its incident angles in $f_{\mathcal E}$ and $g_{\mathcal E}$ are both small. Finally, the first (resp.\ last) vertex of $\mathcal P_i$ has one incident flat angle in $f_j$ (resp.\ by $f_{j-1}$) and one incident flat angle in $f_{\mathcal E}$, hence, by Condition~(UP1), its incident angle in $g_{\mathcal E}$ is small. 

We are now ready to count. All the angles in $f_{\mathcal E}$ at the vertices of $\mathcal P_i$ are labeled~$-1$, except for the angles at the end-vertices of $\mathcal P_i$, which are labeled~$0$; these are the $6a_i$ angles in $f_{\mathcal E}$ labeled $-1$. The vertices $v_j$ and $w_j$ in $\mathcal V_i$ are not incident to $f_{\mathcal E}$. Thus, no angle in $f_{\mathcal E}$ is labeled~$+1$. Further, all the angles in $g_{\mathcal E}$ at the vertices of $\mathcal P_i$ are labeled~$-1$; these are the $6a_i+2$ angles in $g_{\mathcal E}$ labeled $-1$. Finally, all the angles in $g_{\mathcal E}$ at the vertices $v_j$ and $w_j$ in $\mathcal V_i$ are labeled~$+1$; these are the $12a_i+2$ angles in $g_{\mathcal E}$ labeled $+1$, as there are $6a_i+1$ cycles $\mathcal C_j$ such that the edge $(u_j,u_{j+1})$ belongs to $\mathcal P_i$.
\end{proof}

The proof of the following is analogous to the one of Lemmaa~\ref{le:contribution-inside}.

\begin{lemma} \label{le:contribution-outside}
	For any $i\in \{1,\dots,k\}$, suppose that $\mathcal Q^*(i)$ lies outside $\mathcal C$ in $\mathcal E$. Then:
	\begin{itemize}
		\item the number of angles of $f_{\mathcal E}$ labeled $-1$ at the vertices in $\mathcal V_i$ is $6a_i+2$;
		\item the number of angles of $f_{\mathcal E}$ labeled~$+1$ at the vertices in $\mathcal V_i$ is~$12a_i+2$;
		\item the number of angles of $g_{\mathcal E}$ labeled~$-1$ at the vertices in $\mathcal V_i$ is $6a_i$;
		\item the number of angles of $g_{\mathcal E}$ labeled~$+1$ at the vertices in $\mathcal V_i$ is $0$.
	\end{itemize}  
\end{lemma} 

Recall that, if $\mathcal Q^*(i)$ lies inside (outside) $\mathcal C$, then $a_i \in \mathcal S_1$ (resp.\ $a_i \in \mathcal S_2$). Thus, by Lemmata~\ref{le:contribution-inside} and~\ref{le:contribution-outside}, we have that, over all the vertices in $\mathcal V_1\cup \dots \cup \mathcal V_k$, the face $f_{\mathcal E}$ is incident to $\sum_{a_i\in \mathcal S_1} (6a_i) + \sum_{a_i\in \mathcal S_2} (6a_i+2)$ angles labeled $-1$ and to $\sum_{a_i\in \mathcal S_2} (12a_i+2)$ angles labeled $+1$, hence the difference $\delta$ between the number of angles labeled $+1$ in $f_{\mathcal E}$ and the number of angles labeled $-1$ in $f_{\mathcal E}$ at the vertices in $\mathcal V_1\cup \dots \cup \mathcal V_k$ is equal to $6(\sum_{a_i\in \mathcal S_2} a_i - \sum_{a_i\in \mathcal S_1} a_i)$. Assume, for a contradiction, that $\sum_{a_i\in \mathcal S_1} a_i \neq \sum_{a_i\in \mathcal S_2} a_i$. This implies that $\delta$ is larger than or equal to $6$, or smaller than or equal to $-6$. Apart from $\mathcal V_1\cup \dots \cup \mathcal V_k$, the vertex set of $G$ only includes two more vertices, namely $u_0$ and $u_{u_{6n+2k+1}}$. Thus, the difference between the number of angles labeled $+1$ in $f_{\mathcal E}$ and the number of angles labeled $-1$ in $f_{\mathcal E}$ is larger than or equal to $4$, or smaller than or equal to $-4$. However, by Condition~(UP3) of Theorem~\ref{th:upward-conditions}, such a difference is equal to $2$, a contradiction which proves that $\sum_{a_i\in \mathcal S_1}a_i = \sum_{a_i\in \mathcal S_2} a_i$ and concludes the proof of Theorem~\ref{th:partition-to-upward}.   		

\section{Solving Upward Planarity Via Partition} \label{se:from-upward-to-partition}

In this section we prove Theorem~\ref{th:upward-to-partition}. Our algorithm borrows some ideas from the algorithm for testing the rectilinear planarity of an outerplanar graph~\cite{DBLP:journals/comgeo/Frati22}. Throughout the section, we assume that there exists an algorithm with $f(n)$ running time that decides whether a set of positive integers can be partitioned into two sets, so that the sum of the elements in each set is equal to $n/2$. 


\begin{figure}[htb]
	\centering
	\includegraphics[scale=.9,page=5]{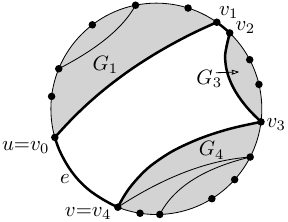}
	\caption{A biconnected outerplanar DAG $G$. Edge directions are not shown. A cycle $\mathcal C$ is thick and the internal faces of the subgraphs $G_i$ of $G$ are colored gray.}
	\label{fig:subgraphs}
\end{figure}

Refer to Fig.~\ref{fig:subgraphs}. Let $G$ be an $n$-vertex biconnected outerplanar DAG, let $\mathcal O$ be any of its two outerplane embeddings, let $e=(u,v)$ be a prescribed edge incident to $f_{\mathcal O}$, and let $\mathcal C=(v_0=u,v_1,\dots,v_k=v)$ be the cycle bounding the internal face of $\mathcal O$ the edge $e$ is incident to. For any $i\in \{1,\dots,k\}$ such that the edge $(v_{i-1},v_i)$ is not incident to $f_{\mathcal O}$, let $\mathcal P_i$ be the path along $f_{\mathcal O}$ that connects $v_{i-1}$ and $v_i$ and that does not comprise vertices of $\mathcal C$ other than $v_{i-1}$ and $v_i$. We denote by $G_i$ the subgraph of $G$ induced by the vertices of $\mathcal P_i$. 

We want to test in $\bigoh(f(n))$ time whether $G$ admits an upward plane embedding in which $e$ is incident to the outer face. Our algorithm not only decides whether $G$ admits such an embedding, but it actually computes the \emph{feasible set} $\mathcal F$ of $G$, that is, the set of all pairs $(\mu,\nu)$ with $\mu,\nu\in\{-1,0,1\}$ such that $G$ admits an upward plane embedding $(\mathcal E,\lambda)$ in which $e$ is incident to $f_{\mathcal E}$ and the angles at $u$ and $v$ inside $\mathcal E$ are $\mu$ and $\nu$, respectively; we call this a \emph{$(\mu,\nu)$-embedding}. As the computation of the feasible sets will eventually be done recursively, we assume to have the feasible set $\mathcal F_i$ of each subgraph $G_i$ of $G$; note that, in a $(\mu_i,\nu_i)$-embedding of $G_i$, the edge $(v_{i-1},v_i)$ is incident to the outer face. Since $\mathcal F$ might comprise at most $9$ pairs, we consider each pair $(\mu,\nu)$ individually and show how to test whether $(\mu,\nu)\in \mathcal F$ in $\bigoh(f(k))$ time. The basis for our algorithm is the following structural lemma.

\begin{figure}[htb]
	\centering
	\includegraphics[scale=.9,page=4]{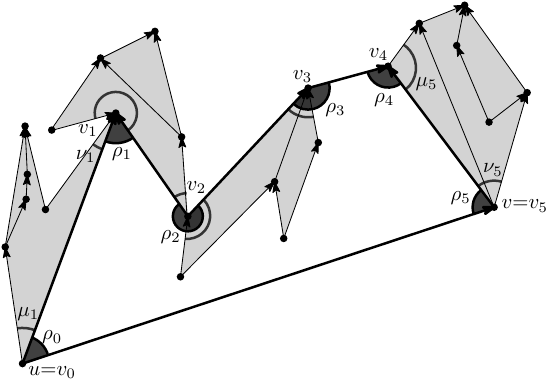}
	\caption{Illustration for Lemma~\ref{le:structure}. The cycle $\mathcal C$ is thick, the internal faces of the subgraphs $G_i$ are gray, the angles $\rho_i$ are black opaque, and the angles $\mu_i$ and $\nu_i$ are not filled.}
	\label{fig:setting}
\end{figure}


\begin{lemma} \label{le:structure}
For any $\mu,\nu\in\{-1,0,1\}$, we have that $G$ admits a $(\mu,\nu)$-embedding if and only if there exist values $\rho_0, \dots,\rho_k,\mu_1,\dots,\mu_k,\nu_1,\dots,\nu_k$ in $\{-1,0,1\}$ and there exists an \emph{in-out assignment} $\mathcal A$, that is, an assignment of each subgraph $G_i$ of $G$ to the inside or to the outside of $\mathcal C$, such that (refer to Fig.~\ref{fig:setting}):
\begin{enumerate} [{$\mathcal P$}1:]
	\item For $i=0,\dots,k$, if $v_i$ is a switch vertex in $\mathcal C$, then $\rho_i\in \{-1,1\}$, else $\rho_i=0$.
	\item $\sum_{i=0}^k \rho_i = -2$.
	\item For $i=1,\dots,k$, if $G_i$ is defined, then $(\mu_i,\nu_i)\in \mathcal F_i$.
	\item For $i=1,\dots,k$, if $G_i$ is defined and $\mathcal A(G_i)=\textrm{in}$, then $\mu_i\leq \rho_{i-1}$ and $\nu_i\leq \rho_{i}$; also, if $G_i$ is defined and $\mathcal A(G_i)=\textrm{out}$, then \mbox{$\mu_i\leq -\rho_{i-1}$ and $\nu_i\leq -\rho_{i}$}.
	\item For $i=1,\dots,k-1$, we have $\nu_i+\mu_{i+1}\leq 0$.
	\item For $i=1,\dots,k-1$, if $\nu_i=\mu_{i+1}=0$, if $G_i$ and $G_{i+1}$ are defined, and if $\mathcal A(G_i)=\mathcal A(G_{i+1})=\textrm{in}$, then $\rho_i=1$; also, if $\nu_i=\mu_{i+1}=0$, if $G_i$ and $G_{i+1}$ are defined, and if $\mathcal A(G_i)=\mathcal A(G_{i+1})=\textrm{out}$, then $\rho_i=-1$.
	\item If $G_1$ is undefined or $\mathcal A(G_1)=\textrm{in}$, then $\mu=\rho_0$, otherwise $\mu=\rho_0+\mu_1+1$. If $G_k$ is undefined or $\mathcal A(G_k)=\textrm{in}$, then $\nu=\rho_k$, otherwise \mbox{$\nu=\rho_k+\nu_k+1$.}
\end{enumerate}
\end{lemma}

\begin{proof}
$(\Longrightarrow)$ Let $(\mathcal E,\lambda)$ be a $(\mu,\nu)$-embedding of $G$. For each $i\in \{1,\dots,k \}$ such that $G_i$ is defined, let $\mathcal A(G_i)=\textrm{in}$ if $G_i$ lies inside $\mathcal C$ in $\mathcal E$, and let $\mathcal A(G_i)=\textrm{out}$ otherwise. For $i=0,\dots,k$, let $\rho_i$ be the angle at $v_i$ inside the upward plane embedding $\mathcal E_{\mathcal C}$ of $\mathcal C$ induced by $\mathcal E$. Also, for $i=1,\dots,k$, if $G_i$ is defined, let $\mu_i$ and $\nu_i$ be the angles at $v_{i-1}$ and $v_i$, respectively, inside the embedding ${\mathcal E}_i$ of $G_i$ induced by $\mathcal E$, otherwise let $\mu_i=\nu_i=-1$.

Properties ${\mathcal P}_1$ and ${\mathcal P}_2$ then follow from Properties (UP0) and (UP3) of Theorem~\ref{th:upward-conditions}, respectively, applied to $\mathcal E_{\mathcal C}$. Property ${\mathcal P}_3$ follows from the definition of $\mu_i$ and $\nu_i$. Property ${\mathcal P}_4$ follows from the fact that, if $G_i$ is defined and $\mathcal A(G_i)=\textrm{in}$, then $G_i$ lies inside $\mathcal C$ in $\mathcal E$, thus the angles at $v_{i-1}$ and $v_i$ inside ${\mathcal E}_i$, which are equal to $\mu_i$ and $\nu_i$, are respectively smaller than or equal to the angles at $v_{i-1}$ and $v_i$ inside $\mathcal E_{\mathcal C}$, which are equal to $\rho_{i-1}$ and $\rho_i$; the case in which $G_i$ is defined and $\mathcal A(G_i)=\textrm{out}$ can be discussed similarly. Property ${\mathcal P}_5$ is trivial if $G_i$ or $G_{i+1}$ is undefined, as in this case one of $\nu_i$ and $\mu_{i+1}$ is equal to $-1$ and the other one does not exceed $1$; otherwise, it comes from Properties (UP1) and (UP2) of Theorem~\ref{th:upward-conditions} -- by restricting $\mathcal E$ to the cycles bounding the outer faces of ${\mathcal E}_i$ and ${\mathcal E}_{i+1}$. Property ${\mathcal P}_6$ comes from the fact that an angle comprising two flat angles is large -- note that ${\mathcal E}_i$ and ${\mathcal E}_{i+1}$ are one outside the other, since $e$ is incident to $f_{\mathcal E}$. Finally, Property ${\mathcal P}_7$ comes from the fact that the angles at $u$ and $v$ inside $\mathcal E$ are $\mu$ and $\nu$, respectively. If $G_1$ is undefined or if it lies inside $\mathcal C$ in $\mathcal E$, then the angle at $u$ inside $\mathcal E$, which is $\mu$, coincides with the angle at $u$ inside $\mathcal E_{\mathcal C}$, which is $\rho_0$. If $G_1$ is defined and lies outside $\mathcal C$ in $\mathcal E$, then the angles at $u$ inside ${\mathcal E}_1$ and inside ${\mathcal E}_{\mathcal C}$ ``sum up'' to compose the angle at $u$ inside ${\mathcal E}$, however the numerical definition of such angles needs a $+1$ in the sum; for example, if the angles at $u$ inside ${\mathcal E}_1$ and inside ${\mathcal E}_{\mathcal C}$ are both equal to $-1$, then the angle at $u$ inside ${\mathcal E}$ is also equal to $-1$; or if they are both equal to $0$, then the angle at $u$ inside ${\mathcal E}$ is equal to $1$. A similar argument applies for the angles at $v$.


$(\Longleftarrow)$ Assume now that there exist values $\rho_0, \dots,\rho_k,\mu_1,\dots,\mu_k,\nu_1,\dots,\nu_k$ in $\{-1,0,1\}$ and there exists an in-out assignment $\mathcal A$ satisfying Properties ${\mathcal P}_1$--${\mathcal P}_7$. We construct a $(\mu,\nu)$-embedding of $G$ as follows. Let $\mathcal E_{\mathcal C}$ be any of the two plane embeddings of $\mathcal C$. For $i=0,\dots,k$, we label the angle at $v_i$ inside $\mathcal E_{\mathcal C}$ as $\rho_i$ and the angle at $v_i$ in $f_{\mathcal E_{\mathcal C}}$ as $-\rho_i$. By Properties ${\mathcal P}_1$ and ${\mathcal P}_2$, the described labeling $\lambda_{\mathcal C}$, together with $\mathcal E_{\mathcal C}$, defines an upward plane embedding of $\mathcal C$. For $i=1,\dots,k$, if $G_i$ is defined, then we let $(\mathcal E_i,\lambda_i)$ be a $(\mu_i,\nu_i)$-embedding of $G_i$. This exists by Property ${\mathcal P}_3$; we embed $\mathcal E_i$ inside $\mathcal C$ if $\mathcal A(G_i)=\textrm{in}$ and we embed $\mathcal E_i$ outside $\mathcal C$ if $\mathcal A(G_i)=\textrm{out}$. This results in a plane embedding $\mathcal E$ of $G$. The assignment $\lambda$ of labels to the angles in $\mathcal E$ is the one obtained by combining $(\mathcal E_{\mathcal C},\lambda_{\mathcal C})$ with the embeddings $(\mathcal E_i,\lambda_i)$. In particular, let $g_{\mathcal E}$ be the internal face of $\mathcal E$ incident to $e$. Then the angle at $v_i$ in $g_{\mathcal E}$ is obtained from $\rho_i$ by subtracting $\nu_i+1$ if $G_i$ is defined and $\mathcal A(G_i)=\textrm{in}$, and by subtracting $\mu_{i+1}+1$ if $G_{i+1}$ is defined and $\mathcal A(G_{i+1})=\textrm{in}$. The angle at $v_i$ in $f_{\mathcal E}$ is determined analogously. Properties ${\mathcal P}_4$--${\mathcal P}_6$ ensure that $(\mathcal E,\lambda)$ is an upward plane embedding. In particular, Property ${\mathcal P}_4$ ensures that the angles at $v_{i-1}$ and $v_i$ inside $\mathcal E_i$ ``fit'' inside $\mathcal E_{\mathcal C}$ (if $\mathcal A(G_i)=\textrm{in}$) or outside $\mathcal E_{\mathcal C}$ (if $\mathcal A(G_i)=\textrm{out}$). Furthermore, Property ${\mathcal P}_5$ ensures that the angles at $v_i$ inside $\mathcal E_i$ and $\mathcal E_{i+1}$ fit together in $\mathcal E$ (avoiding the possibility that they are one large and one flat, or both large). Moreover, Property ${\mathcal P}_6$ ensures that the angles at $v_i$ inside $\mathcal E_i$ and $\mathcal E_{i+1}$ simultaneously fit inside $\mathcal E_{\mathcal C}$ (if $G_i$ and $G_{i+1}$ are defined and $\mathcal A(G_i)=\mathcal A(G_{i+1})=\textrm{in}$) or outside $\mathcal E_{\mathcal C}$ (if $G_i$ and $G_{i+1}$ are defined and $\mathcal A(G_i)=\mathcal A(G_{i+1})=\textrm{out}$). Finally, Property ${\mathcal P}_7$ ensures that the upward plane embedding $(\mathcal E,\lambda)$ is a $(\mu,\nu)$-embedding, by requiring that the angles at $u$ and $v$ inside $\mathcal E$ are $\mu$ and $\nu$, respectively.
\end{proof}

By Lemma~\ref{le:structure}, in order to test whether a pair $(\mu,\nu)$ with $\mu,\nu\in\{-1,0,1\}$ belongs to the feasible set $\mathcal F$ of $G$, our algorithm will decide whether there exist values $\rho_0, \dots,\rho_k$, $\mu_1,\dots,\mu_k$, $\nu_1,\dots,\nu_k$ in $\{-1,0,1\}$ and an in-out assignment $\mathcal A$ satisfying Properties ${\mathcal P}_1$--${\mathcal P}_7$. 

The overview of our algorithm is as follows. First, we fix $\bigoh(1)$ values among the ones that we have to decide, in all possible ways. This ensures that things are ``fine'' close to $u$ and $v$, at the expense of performing the next steps $\bigoh(1)$ times, namely once for each way to fix the above $\bigoh(1)$ values. Then we show that the remaining values $\mu_i$ and $\nu_i$ can be decided in $\bigoh(k)$ time without loss of generality by looking at the graph structure. Finally, we have to decide the values $\rho_i$, as well as the in-out assignment $\mathcal A$. The previous decisions on the values $\mu_i$ and $\nu_i$ bind together some subgraphs $G_i$; assigning such bound subgraphs to the inside or to the outside of $\mathcal C$ gives a contribution to the sum of the angles inside or outside $\mathcal C$. Balancing such contributions can be done via \partition.

Recall that we have, for each subgraph $G_i$ of $G$, its feasible set $\mathcal F_i$. We can assume that $\mathcal F_i\neq \emptyset$, as otherwise we can conclude that $G$ admits no $(\mu,\nu)$-embedding, by Property~${\mathcal P}_3$ of Lemma~\ref{le:structure}. The following lemma will be useful.

\begin{lemma} \label{le:simultaneous-values}
A feasible set $\mathcal F_i$ does not contain two pairs $(-1,\cdot)$ and $(0,\cdot)$, and does not contain two pairs $(\cdot,-1)$ and $(\cdot,0)$.
\end{lemma}	

\begin{proof}
If $\mathcal F_i$ contains a pair $(-1,\cdot)$, then $v_{i-1}$ is a switch vertex in $G_i$, while if it contains a pair $(0,\cdot)$, then $v_{i-1}$ is a non-switch vertex in $G_i$, hence $\mathcal F_i$ cannot contain both pairs. The proof for the pairs $(\cdot,-1)$ and $(\cdot,0)$ is analogous.
\end{proof}

We start by fixing $\rho_0$ and $\rho_k$ in every possible way -- each of them has a value in $\{-1,0,1\}$. Also, if $G_1$ is defined, we fix $\mu_1$ and $\mathcal A(G_1)$ in every possible way -- note that $\mu_1\in \{-1,0,1\}$ and $\mathcal A(G_1)\in \{\textrm{in},\textrm{out}\}$; otherwise, we let $\mu_1=-1$. Analogously, if $G_k$ is defined, we fix $\mu_k$, $\nu_k$, and $\mathcal A(G_k)$ in every possible way, otherwise we let $\mu_k=\nu_k=-1$. This results in $\bigoh(1)$ tuples of values $\langle \rho_0, \rho_k, \mu_1, \mathcal A(G_1), \mu_k, \nu_k, \mathcal A(G_k)\rangle$. For each of them, we test in $\bigoh(1)$ time whether it satisfies Properties ${\mathcal P}_1$ (this concerns the values $\rho_0$ and $\rho_k$),  ${\mathcal P}_3$ (this concerns the values $\mu_k$ and $\nu_k$), ${\mathcal P}_4$ (this concerns the values $\rho_0$, $\mu_1$ and $\mathcal A(G_1)$, and the values $\rho_k$, $\nu_k$ and $\mathcal A(G_k)$), and ${\mathcal P}_7$ (this concerns the values $\rho_0$, $\mu_1$ and $\mathcal A(G_1)$, and the values $\rho_k$, $\nu_k$ and $\mathcal A(G_k)$) of Lemma~\ref{le:structure}. If not, we discard the tuple. Otherwise, we say it is a \emph{candidate tuple}. We are going to independently test in $\bigoh(f(k))$ time whether each candidate tuple is \emph{extensible}, i.e., whether there exist values for $\rho_1, \dots,\rho_{k-1},\mu_2,\dots,\mu_{k-1},\nu_1,\dots,\nu_{k-1}$ and assignments for the subgraphs $G_2,\dots,G_{k-1}$ that are defined that, together with the values in the candidate tuple,  satisfy Properties ${\mathcal P}_1$--${\mathcal P}_7$ of Lemma~\ref{le:structure}. 

The following lemmata prove that the values $\mu_1,\nu_1,\mu_2,\dots,\nu_{k-1}$ can be decided {\em without loss of generality}, in this order; this means that, if we decided the values $\mu_1,\nu_1,\mu_2,\nu_2,\dots,\mu_i$ (or $\nu_1,\mu_2,\nu_2,\dots,\nu_i$), then the candidate tuple under consideration is extensible if and only if an extension exists with the chosen values $\nu_1,\mu_2,\nu_2,\dots,\mu_i$ (resp., $\nu_1,\mu_2,\nu_2,\dots,\nu_i$).

\begin{lemma} \label{le:from-mu-to-nu}
Suppose that, for some $i\in\{1,\dots,k-1\}$, the values $\mu_1,\nu_1,\dots,\mu_i$ have been decided without loss of generality. Then $\nu_i$ can also be decided without loss of generality in $\bigoh(1)$ time, as follows. If $G_i$ is undefined, we let $\nu_i=-1$. Otherwise, if $\mathcal F_i$ does not contain any pair $(\mu_i,\cdot)$, then it is concluded that the current candidate tuple is not extensible. Otherwise, we let $\nu_i$ be the smallest value $\beta$ such that $(\mu_i,\beta)\in \mathcal F_i$. 
\end{lemma}

\begin{proof}
First, deciding $\nu_i$ takes $\bigoh(1)$ time as the size of $\mathcal F_i$ is $\bigoh(1)$. In order to prove that $\nu_i$ is decided without loss of generality, we only discuss the case in which $G_i$ is defined, as the case in which it is not is simpler. 

Since $\mu_1,\nu_1,\dots,\mu_i$ have been decided without loss of generality, if the candidate tuple $\langle \rho_0, \rho_k$, $\mu_1, \mathcal A(G_1)$, $\nu_k$, $\mathcal A(G_k)\rangle$ is extensible, an extension exists with the chosen values $\nu_1,\mu_2,\nu_2,\dots,\mu_i$. If $\mathcal F_i$ does not contain any pair $(\mu_i,\cdot)$, then because the value $\mu_i$ has been decided without loss of generality, we can conclude that that the current candidate tuple is not extensible. Otherwise, suppose that an extension exists and that, in such an extension, $\nu_i$ is not the smallest value $\beta$ such that $(\mu_i,\beta)\in \mathcal F_i$, say $\nu_i=\beta'\neq \beta$. By Property~${\mathcal P}_3$ and by definition of feasible set, we have $(\mu_i,\beta')\in \mathcal F_i$, which implies that $\beta'>\beta$. Then, in the extension of $\langle \rho_0, \rho_k, \mu_1, \mathcal A(G_1), \nu_k, \mathcal A(G_k)\rangle$, we can replace the value $\nu_i=\beta'$ with $\nu_i=\beta$, thus decreasing such a value and leaving unaltered all other values and the in-out assignment. The resulting set of values, together with the in-out assignment, still satisfies Properties ${\mathcal P}_1$--${\mathcal P}_7$ of Lemma~\ref{le:structure}. Indeed, Properties ${\mathcal P}_1$, ${\mathcal P}_2$, and ${\mathcal P}_7$ are not affected by the change of value for $\nu_i$ (note that $i\leq k-1$). Property~${\mathcal P}_3$ is satisfied as  $(\mu_i,\beta)\in \mathcal F_i$. The inequalities $\nu_i\leq \rho_{i}$, $\nu_i\leq -\rho_{i}$, and $\nu_i+\mu_{i+1}\leq 0$ of Properties ${\mathcal P}_4$ and ${\mathcal P}_5$ are satisfied with $\nu_i=\beta'$, hence they are also satisfied with $\nu_i=\beta<\beta'$. Finally, it cannot be that $\nu_i=\beta=\mu_{i+1}=0$, as this would imply that $\beta'=1$, and hence that Property~${\mathcal P}_5$ is not satisfied in the given extension; it follows that Property~${\mathcal P}_6$ is also satisfied. 
\end{proof}

\begin{lemma} \label{le:from-nu-to-mu}
	Suppose that, for some $i\in\{1,\dots,k-2\}$, the values $\mu_1,\nu_1,\dots,\nu_i$ have been decided without loss of generality. Then $\mu_{i+1}$ can also be decided without loss of generality in $\bigoh(1)$ time as follows. If $G_{i+1}$ is undefined, we let $\mu_{i+1}=-1$. Otherwise, we distinguish four cases.
	\begin{itemize}
		\item If $\nu_{i}=1$, then $\mu_{i+1}=-1$; if $\mathcal F_{i+1}$ does not contain any pair $(-1,\cdot)$, then it is concluded that the current candidate tuple is not extensible.
		\item If $\nu_{i}=0$ and $\mathcal F_{i+1}$ contains a pair $(-1,\cdot)$ or contains a pair $(0,\cdot)$, then we let $\mu_{i+1}=-1$ or $\mu_{i+1}=0$, respectively. If $\nu_{i}=0$ and $\mathcal F_{i+1}$ contains neither a pair $(-1,\cdot)$ nor a pair $(0,\cdot)$, then it is concluded that the current candidate tuple is not extensible.
		\item Suppose that $\nu_{i}=-1$ and $v_i$ is a non-switch vertex in $\mathcal C$. If $\mathcal F_{i+1}$ contains a pair $(-1,\cdot)$ or contains a pair $(0,\cdot)$, then we let $\mu_{i+1}=-1$ or $\mu_{i+1}=0$, respectively. If $\mathcal F_{i+1}$ contains neither a pair $(-1,\cdot)$ nor a pair $(0,\cdot)$, then it is concluded that the current candidate tuple is not extensible.
		\item Suppose finally that $\nu_{i}=-1$ and $v_i$ is a switch vertex in $\mathcal C$. 
		\begin{itemize}
			\item If $\mathcal F_{i+1}$ contains a pair $(\alpha,\beta)$ with $\{\alpha,\beta\}\subseteq \{-1,0\}$, we let $\mu_{i+1}=\alpha$. 
			\item Otherwise, if $\mathcal F_{i+1}$ contains the pair $(1,-1)$, we let $\mu_{i+1}=1$.
			\item Otherwise, if $\mathcal F_{i+1}$ contains the pair $(1,0)$ and does not contain the pair $(-1,1)$, we let $\mu_{i+1}=1$.
			\item Otherwise, if $\mathcal F_{i+1}$ contains the pair $(\alpha,1)$ with $\alpha \in \{-1,0\}$ and does not contain the pair $(1,0)$, we let $\mu_{i+1}=\alpha$.
			\item Otherwise, if $\mathcal F_{i+1}$ contains the pairs $(1,0)$ and $(-1,1)$, we look at $v_{i+1}$ and $G_{i+2}$. If $v_{i+1}$ is a non-switch vertex in $\mathcal C$, we let $\mu_{i+1}=1$. Otherwise, if $G_{i+2}$ is defined and $\mathcal F_{i+2}$ contains a pair $(0,\cdot)$, we let $\mu_{i+1}=1$. Otherwise, we let $\mu_{i+1}=-1$. 
			\item Otherwise, $\mathcal F_{i+1}=\{(1,1)\}$, and we let $\mu_{i+1}=1$. 
		\end{itemize} 
	\end{itemize}  
\end{lemma}

\begin{proof}
First, deciding $\mu_{i+1}$ takes $\bigoh(1)$ time as the size of $\mathcal F_{i+1}$ is $\bigoh(1)$. In order to prove that $\mu_{i+1}$ is decided without loss of generality, we only discuss the case in which $G_{i+1}$ is defined, as the case in which it is not is simpler.	Since $\mu_1,\nu_1,\dots,\nu_i$ have been decided without loss of generality, if there exists an extension of the current candidate tuple, there exists an extension $\varphi$ in which $\mu_1,\nu_1,\dots,\nu_i$ have the decided values. 

Suppose first that $\nu_{i}=1$. Property~${\mathcal P}_5$ then implies that $\mu_{i+1}=-1$ in $\varphi$, which coincides with our decision. Property~${\mathcal P}_3$, together with the fact that $\mu_1,\nu_1,\dots,\nu_i$ have been decided without loss of generality, implies that the current candidate tuple is not extensible if $\mathcal F_{i+1}$ does not contain any pair $(-1,\cdot)$. 

Suppose next that $\nu_{i}=0$. Property~${\mathcal P}_5$ then implies that $\mu_{i+1}=-1$ or $\mu_{i+1}=0$. By Lemma~\ref{le:simultaneous-values}, we have that $\mathcal F_{i+1}$ cannot contain both a pair $(-1,\cdot)$ and a pair $(0,\cdot)$. Hence, also in this case, our decision for the value of $\mu_{i+1}$ coincides with the one in $\varphi$. Property~${\mathcal P}_3$, together with the fact that $\mu_1,\nu_1,\dots,\nu_i$ have been decided without loss of generality, implies that the current candidate tuple is not extensible if $\mathcal F_{i+1}$ does not contain any pair $(-1,\cdot)$ or $(0,\cdot)$.

Suppose next that $\nu_{i}=-1$ and $v_i$ is a non-switch vertex in $\mathcal C$. By Property~${\mathcal P}_1$, we have $\rho_i=0$. By Property~${\mathcal P}_4$, we have $\mu_{i+1}\leq 0$, hence $\mu_{i+1}=-1$ or $\mu_{i+1}=0$. By Lemma~\ref{le:simultaneous-values}, we have that $\mathcal F_{i+1}$ cannot contain both a pair $(-1,\cdot)$ and a pair $(0,\cdot)$. Hence, also in this case, our decision for the value of $\mu_{i+1}$ coincides with the one in $\varphi$. Property~${\mathcal P}_3$, together with the fact that $\mu_1,\nu_1,\dots,\nu_i$ have been decided without loss of generality, implies that the current candidate tuple is not extensible if $\mathcal F_{i+1}$ does not contain any pair $(-1,\cdot)$ or $(0,\cdot)$.

Suppose finally that $\nu_{i}=-1$ and $v_i$ is a switch vertex in $\mathcal C$. By Property~${\mathcal P}_1$, we have $\rho_i=-1$ or $\rho_i=1$ in $\varphi$; suppose the former, as the other case is analogous. Also, let $\alpha'$ and $\beta'$ be the values of $\mu_{i+1}$ and $\nu_{i+1}$ in $\varphi$. We distinguish some cases.
\begin{itemize}
	\item Suppose that $\mathcal F_{i+1}$ contains a pair $(\alpha,\beta)$ with $\{\alpha,\beta\}\subseteq \{-1,0\}$. By Lemma~\ref{le:simultaneous-values}, we have that $\mathcal F_{i+1}$ contains only one such a pair, and that $\alpha\leq \alpha'$ and $\beta\leq \beta'$. Then, in $\varphi$, we can replace the value $\mu_{i+1}=\alpha'$ and $\nu_{i+1}=\beta'$ with $\mu_{i+1}=\alpha$ and $\nu_{i+1}=\beta$, leaving unaltered all other values and the in-out assignment. The resulting set of values, together with the in-out assignment, still satisfies Properties ${\mathcal P}_1$--${\mathcal P}_7$ of Lemma~\ref{le:structure}. Indeed, Properties ${\mathcal P}_1$, ${\mathcal P}_2$, and ${\mathcal P}_7$ are not affected by the change of value for $\mu_{i+1}$ and $\nu_{i+1}$ (note that $i\leq k-2$). Property~${\mathcal P}_3$ is satisfied as $(\alpha,\beta)\in \mathcal F_{i+1}$. The inequalities $\mu_{i+1}\leq \rho_{i}$, $\nu_{i+1}\leq \rho_{i+1}$, $\mu_{i+1}\leq -\rho_{i}$, $\nu_{i+1}\leq -\rho_{i+1}$, $\nu_{i}+\mu_{i+1}\leq 0$, and $\nu_{i+1}+\mu_{i+2}\leq 0$ of Properties ${\mathcal P}_4$ and ${\mathcal P}_5$ are satisfied with $\mu_{i+1}=\alpha'$ and $\nu_{i+1}=\beta'$, hence they are also satisfied with $\mu_{i+1}=\alpha\leq \alpha'$ and $\nu_{i+1}=\beta\leq \beta'$. Finally, if $\beta=\nu_{i+1}=\mu_{i+2}=0$, then Property~${\mathcal P}_5$ and $\beta\leq \beta'$ imply that $\beta'$ is also equal to $0$, hence  Property~${\mathcal P}_6$ is satisfied by the new values as it is satisfied by $\varphi$. 
	\item Suppose that the previous case does not apply and that $\mathcal F_{i+1}$ contains the pair $(1,-1)$. Then, in $\varphi$, we can replace the value $\mu_{i+1}=\alpha'$ and $\nu_{i+1}=\beta'$ with $\mu_{i+1}=1$ and $\nu_{i+1}=-1$, leaving unaltered all other values and the in-out assignment, except for $G_{i+1}$ which is now assigned to $\textrm{out}$ -- previously it might have been assigned to $\textrm{in}$ or to $\textrm{out}$. The resulting set of values, together with the in-out assignment, satisfies Properties ${\mathcal P}_1$--${\mathcal P}_7$ of Lemma~\ref{le:structure}. In particular, Properties ${\mathcal P}_1$, ${\mathcal P}_2$, and ${\mathcal P}_7$ are not affected by the change of value for $\mu_{i+1}$ and $\nu_{i+1}$, while Property~${\mathcal P}_3$ is satisfied as $(1,-1)\in \mathcal F_{i+1}$. Property ${\mathcal P}_4$ is satisfied since we have $\mu_{i+1}=1=-\rho_{i}$ and $\nu_{i+1}=-1\leq -\rho_{i+1}$. Property ${\mathcal P}_5$ is satisfied since we have $\nu_i+\mu_{i+1}=-1+1=0$ and $\nu_{i+1}+\mu_{i+2}=-1+\mu_{i+2}\leq 0$. Finally, Property ${\mathcal P}_6$ is vacuously satisfied, since $\mu_{i+1}\neq 0$ and $\nu_{i+1}\neq 0$.
	\item Suppose that the previous cases do not apply, that $\mathcal F_{i+1}$ contains the pair $(1,0)$, and that $\mathcal F_{i+1}$ does not contain the pair $(-1,1)$. Then, in $\varphi$, we can replace the value $\mu_{i+1}=\alpha'$ and $\nu_{i+1}=\beta'$ with $\mu_{i+1}=1$ and $\nu_{i+1}=0$, leaving unaltered all other values and the in-out assignment. The resulting set of values, together with the in-out assignment, satisfies Properties ${\mathcal P}_1$--${\mathcal P}_7$ of Lemma~\ref{le:structure}. In particular, Properties ${\mathcal P}_1$, ${\mathcal P}_2$, and ${\mathcal P}_7$ are not affected by the change of value for $\mu_{i+1}$ and $\nu_{i+1}$, while Property~${\mathcal P}_3$ is satisfied as $(1,0)\in \mathcal F_{i+1}$. We now discuss Property ${\mathcal P}_4$. Because of the assumptions of this case, we have that $\mathcal F_{i+1}$ contains $(1,0)$, might contain $(0,1)$ and $(1,1)$, and does not contain any other pair. Since $\varphi$ satisfies Property ${\mathcal P}_3$, we have $\alpha'\geq 0$ and $\beta'\geq 0$. Since $\varphi$ satisfies Property ${\mathcal P}_4$, we have $\mathcal A(G_{i+1})=\textrm{out}$, hence $\mu_{i+1}=1=-\rho_i$ and $\nu_{i+1}=0\leq \beta' \leq -\rho_{i+1}$. Property ${\mathcal P}_5$ is satisfied since we have $\nu_i+\mu_{i+1}=-1+1=0$ and $\nu_{i+1}+\mu_{i+2}=0+\mu_{i+2}\leq \beta'+\mu_{i+2}\leq 0$, given that $\varphi$ satisfies Property ${\mathcal P}_5$. Finally, $\mu_{i+1}\neq 0$; further, if $\beta=\nu_{i+1}=\mu_{i+2}=0$, then Property~${\mathcal P}_5$ and $\beta\leq \beta'$ imply that $\beta'$ is also equal to $0$, hence  Property~${\mathcal P}_6$ is satisfied as it is satisfied by $\varphi$. 
	\item Suppose that the previous cases do not apply, that $\mathcal F_{i+1}$ contains the pair $(\alpha,1)$, with $\alpha \in \{-1,0\}$, and that $\mathcal F_{i+1}$ does not contain the pair $(1,0)$. By Lemma~\ref{le:simultaneous-values}, we have that $\mathcal F_{i+1}$ contains only one such a pair, and that $\alpha\leq \alpha'$. Because of the assumptions of this case, we have that $\mathcal F_{i+1}$ contains $(\alpha,1)$, might contain $(1,1)$, and does not contain any other pair. Hence, we have $\beta'=1$. Then, in $\varphi$, we can replace the value $\mu_{i+1}=\alpha'$ with $\mu_{i+1}=\alpha$, leaving unaltered all other values and the in-out assignment. Since $\mu_{i+1}$ does not increase, the proof that the resulting set of values, together with the in-out assignment, satisfies Properties ${\mathcal P}_1$--${\mathcal P}_7$ of Lemma~\ref{le:structure} is the same as in the first case.
	\item Suppose that the previous cases do not apply, and that $\mathcal F_{i+1}$ contains the pairs $(1,0)$ and $(-1,1)$. Because of the assumptions of this case, we have that $\mathcal F_{i+1}$ might contain $(0,1)$, $(1,1)$, and does not contain any other pair. Hence, $\beta'\geq 0$. We further distinguish three sub-cases.
	\begin{itemize}
		\item First, suppose that $v_{i+1}$ is a non-switch vertex in $\mathcal C$. By Property~${\mathcal P}_1$, we have $\rho_{i+1}=0$ in $\varphi$. By Property~${\mathcal P}_4$, we have that $\nu_{i+1}\leq 0$. Hence, we have $\alpha'=1$ and $\beta'=0$, thus we can pick the value $1$ for $\mu_{i+1}$ without loss of generality.
		\item Second, suppose that $v_{i+1}$ is a switch vertex in $\mathcal C$, that $G_{i+2}$ is defined, and that $\mathcal F_{i+2}$ contains a pair $(0,\cdot)$. By Lemma~\ref{le:simultaneous-values}, we have that $\mathcal F_{i+2}$ does not contain a pair $(-1,\cdot)$, hence $\mu_{i+2}\geq 0$ in $\varphi$.  By Property~${\mathcal P}_5$, we have that $\nu_{i+1}\leq 0$. Hence, we have $\alpha'=1$ and $\beta'=0$, thus we can pick the value $1$ for $\mu_{i+1}$ without loss of generality.
		\item Third, suppose that $v_{i+1}$ is a switch vertex in $\mathcal C$, and that $G_{i+2}$ is undefined or that $\mathcal F_{i+2}$ does not contain a pair $(0,\cdot)$. First, we observe that $\mu_{i+2}=-1$ in $\varphi$; indeed, $\mu_{i+2}=1$, together with $\nu_{i+1}=\beta'\geq 0$, would violate Property~${\mathcal P}_5$ for $\varphi$. Then, in $\varphi$, we can replace the value $\mu_{i+1}=\alpha'$ and $\nu_{i+1}=\beta'$ with $\mu_{i+1}=-1$ and $\nu_{i+1}=1$, leaving unaltered all other values and the in-out assignment. The resulting set of values, together with the in-out assignment, satisfies Properties ${\mathcal P}_1$--${\mathcal P}_7$ of Lemma~\ref{le:structure}. In particular, Properties ${\mathcal P}_1$, ${\mathcal P}_2$, and ${\mathcal P}_7$ are not affected by the change of value for $\mu_{i+1}$ and $\nu_{i+1}$, while Property~${\mathcal P}_3$ is satisfied as $(-1,1)\in \mathcal F_{i+1}$. We now discuss Property ${\mathcal P}_4$. Since $\beta'\geq 0$, we have that $\rho_{i+1}=1$ if $\mathcal A(G_{i+1})=\textrm{in}$ and $\rho_{i+1}=-1$ if $\mathcal A(G_{i+1})=\textrm{out}$; since $\mathcal A(G_{i+1})$ does not change, we have that the new value $\nu_{i+1}=1$ is equal to $\rho_{i+1}$ if $\mathcal A(G_{i+1})=\textrm{in}$ and to $-\rho_{i+1}$ if $\mathcal A(G_{i+1})=\textrm{out}$. Property ${\mathcal P}_5$ is satisfied since we have $\nu_i+\mu_{i+1}=-1-1<0$ and $\nu_{i+1}+\mu_{i+2}=1-1=0$. Finally, Property ${\mathcal P}_6$ is vacuously satisfied, since $\mu_{i+1}\neq 0$ and $\nu_{i+1}\neq 0$.
	\end{itemize} 

	\item Finally, if neither of the previous cases applies, $\mathcal F_{i+1}$ can only contain the pair $(1,1)$, and in fact it has to contain it, given that $\mathcal F_{i+1}$ is not empty. By Property~${\mathcal P}_3$, we can pick the value $1$ for $\mu_{i+1}$ without loss of generality.
\end{itemize}
This concludes the proof of the lemma.
\end{proof}

By means of Lemmata~\ref{le:from-mu-to-nu} and~\ref{le:from-nu-to-mu}, we can assume to have decided without loss of generality the values $\mu_1,\dots,\mu_k,\nu_1,\dots,\nu_k$. Actually, when the value of $\nu_{k-1}$ is decided, we also need to check whether Property~${\mathcal P}_5$  of Lemma~\ref{le:structure} is satisfied by $\nu_{k-1}$ and $\mu_k$, as the value of $\mu_k$ was decided when fixing the current candidate tuple. The procedure described in Lemmata~\ref{le:from-mu-to-nu} and~\ref{le:from-nu-to-mu} takes $\bigoh(k)$ time, namely $\bigoh(1)$ time per fixed value. By Property~${\mathcal P}_1$, we can also fix the value $\rho_i=0$ for every vertex $v_i$ which is a non-switch vertex in $\mathcal C$. It remains to fix the value $\rho_i$ for every vertex $v_i$ which is a switch vertex in $\mathcal C$ (to $-1$ or $+1$, again by Property~${\mathcal P}_1$), and to decide the in-out assignment $\mathcal A$. The in-out assignments of some components are bound together, as in the following lemma. 

\begin{lemma} \label{le:binding}
Suppose that, for some $i\in \{1,\dots,k-1\}$, the values $\nu_i$ and $\mu_{i+1}$ have been both fixed to $0$. Then, if $v_i$ is a switch vertex in $\mathcal C$ we have $\mathcal A(G_i)=\mathcal A(G_{i+1})$, otherwise we have $\mathcal A(G_i)\neq \mathcal A(G_{i+1})$.  	
\end{lemma}

\begin{proof}
The first part of the statement comes from the fact that if $v_i$ is a switch vertex in $\mathcal C$, then $\{\rho_i,-\rho_i\} =\{-1,1\}$, hence Property~${\mathcal P}_4$ of Lemma~\ref{le:structure} is violated if $G_i$ and $G_{i+1}$ are assigned to different sides of $\mathcal C$. The second part of the statement comes from the fact that if $v_i$ is a non-switch vertex in $\mathcal C$, then $\rho_i=-\rho_i=0$, hence Property~${\mathcal P}_6$ is violated if $G_i$ and $G_{i+1}$ are assigned to the same side of $\mathcal C$.
\end{proof}	

If we consider a maximal sequence of components $G_p,G_{p+1},\dots,G_q$ such that any two consecutive components are bound together, deciding whether $G_p$ is embedded inside or outside $\mathcal C$ decides whether each of the components $G_p,G_{p+1},\dots,G_q$ is embedded inside or outside $\mathcal C$, via repeated applications of Lemma~\ref{le:binding}. 

Furthermore, the values $\rho_{p},\rho_{p+1},\dots,\rho_{q-1}$ are also determined by the same decision, by Property~${\mathcal P}_6$: If $G_i$ and $G_{i+1}$ both lie inside $\mathcal C$, then $\rho_i=1$, while if they both lie outside $\mathcal C$, then $\rho_i=-1$ (while if they lie one inside and one outside $\mathcal C$, then $v_i$ is a switch vertex in $\mathcal C$ and $\rho_i$ has been already fixed to $0$). The values $\rho_{p-1}$ and $\rho_{q}$ might also be determined by the choice of whether $G_p$ is embedded inside or outside $\mathcal C$, if $\mu_p\geq 0$ and if $\nu_q\geq 0$, respectively, by Property~${\mathcal P}_4$. Thus, each maximal sequence of components $G_p,G_{p+1},\dots,G_q$ such that any two consecutive components are bound together gives us a (not necessarily positive) integer, which is the sum of all the $\rho_i$'s that are determined by the choice of whether $G_p$ is embedded inside or outside $\mathcal C$. Note that distinct maximal sequences of components are independently embedded inside or outside $\mathcal C$. By Property~${\mathcal P}_2$, the sum of the positive integers that are obtained in this way has to coincide with the sum of the negative integers, up to a constant additive term. This is where \partition shows up to help.

\begin{figure}[htb]
	\centering
	\includegraphics{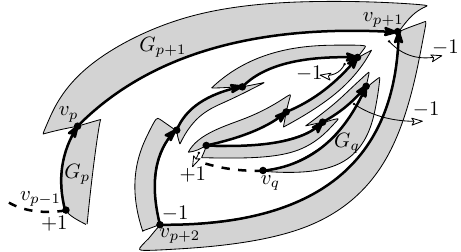}
	\caption{Generation of an integer $a_p$ from components $G_p,G_{p+1},\dots,G_q$ that are bound together. The thick line is part of the cycle $\mathcal C$. Subgraphs $G_i$ are schematized. Contributions of the values $\rho_i$ to $a_p$ are shown. In this example, $a_p=-2$. Note that, in this example, the value of $a_p$ is increased by one unit by $\rho_{p-1}$, since $v_{p-1}$ is a switch vertex in $\mathcal C$ and $\mu_p\geq 0$, while the value of $a_p$ is not altered by $\rho_{q}$, since $\nu_q=-1$.}
	\label{fig:number-generation}
\end{figure}

We now refine this argument; see Fig.~\ref{fig:number-generation}. We are going to define a multiset $\mathcal S$ of positive integers. Initially, we have $\mathcal S=\emptyset$. Consider every maximal sequence of components $G_p,G_{p+1},\dots,G_q$ such that, for $i=p,p+1,\dots,q-1$, the values $\nu_i$ and $\mu_{i+1}$ have both been fixed to $0$. Initialize an integer $a_p$ to $0$ and a variable $c$ to $\textrm{in}$; this variable only assumes values in $\{\textrm{in},\textrm{out}\}$. If $\mu_{p}>=0$ and $v_{p-1}$ is a switch vertex in $\mathcal C$, then increase $a_p$ by $1$ and \emph{mark} $v_{p-1}$ -- thus taking note of the fact that the value $\rho_{p-1}$ has already contributed to generate some integer. Then, for $i=p,\dots,q-1$, proceed as follows. If $v_i$ is a non-switch vertex in $\mathcal C$, then change the value of $c$. If $v_i$ is a switch vertex in $\mathcal C$, then increase the value of $a_p$ by $1$ if $c=\textrm{in}$ and decrease the value of $a_p$ by $1$ if $c=\textrm{out}$; in both cases, mark~$v_i$. Finally, if $\nu_q\geq 0$ and $v_q$ is a switch vertex in $\mathcal C$, then increase the value of $a_p$ by $1$ if $c=\textrm{in}$ and decrease the value of $a_p$ by $1$ if $c=\textrm{out}$; in both cases, mark~$v_q$. Insert into $\mathcal S$ the absolute value of~$a_p$. Henceforth, this process inserts a positive integer into $\mathcal S$ for every maximal sequence of components $G_p,G_{p+1},\dots,G_q$ that are bound together  -- this is true also for sequences consisting of a single component $G_p$. We also insert into $\mathcal S$ a number $1$ for each unmarked switch vertex in $\mathcal C$. In the resulting multiset $\mathcal S$ of integers, there are up to four \emph{special} integers, the ones corresponding to $v_0$, $G_1$, $v_k$, and $G_k$. These are special because, in the current candidate tuple, the values of $\rho_0$ and $\rho_k$ are already decided, as well as the in-out assignment for $G_1$ and $G_k$. The special integers might be less than four. For example, if $v_0$ is a non-switch vertex in $\mathcal C$, it does not generate any integer; as another example, if $v_0$ is marked when considering $G_1$, then their corresponding integer is the same. 

Let $2\ell\leq k$ be the number of switch vertices in $\mathcal C$. Property~${\mathcal P}_2$ of Lemma~\ref{le:structure} requires that the sum of the $\rho_i$'s is equal to $-2$. Because of the above construction, this is possible if and only if the integers in $\mathcal S$ can be partitioned into two sets $\mathcal S_{\textrm{in}}$ and $\mathcal S_{\textrm{out}}$ such that the sum of the elements in $\mathcal S_{\textrm{in}}$ is $\ell-1$ and the sum of the elements in $\mathcal S_{\textrm{out}}$ is $\ell+1$. Special integers are removed from $\mathcal S$ and used to decrease the value $\ell-1$ or $\ell+1$. For example, if $\rho_0=1$ in the current candidate tuple and $v_0$ is not marked, then the integer $a_0=1$ is required to be assigned to $\mathcal S_{\textrm{in}}$, hence we delete it from $\mathcal S$ and we decrease the target sum for the elements in  $\mathcal S_{\textrm{in}}$ to $\ell-2$. After the $\bigoh(1)$ special integers are considered, we are left with a set of integers $\mathcal S$ which needs to be partitioned into two sets $\mathcal S_{\textrm{in}}$ and $\mathcal S_{\textrm{out}}$ whose elements need to sum up to some values $\ell_{\textrm{in}}$ and $\ell_{\textrm{out}}$, respectively. If $\ell_{\textrm{in}}\neq \ell_{\textrm{out}}$, this is not exactly an instance of \partition, hence consider a new integer $a^*=|\ell_{\textrm{in}}-\ell_{\textrm{out}}|$ and we aim to partition $\mathcal S\cup \{a^*\}$ into two sets so that the sum of the elements in each set is the same. Clearly, this is possible if and only if $\mathcal S$ can be partitioned into two sets $\mathcal S_{\textrm{in}}$ and $\mathcal S_{\textrm{out}}$ such that the sum of the elements in $\mathcal S_{\textrm{in}}$ is $\ell_{\textrm{in}}$ and the sum of the elements in $\mathcal S_{\textrm{out}}$ is $\ell_{\textrm{out}}$. Thus, we can solve the problem using the $\bigoh(f(\ell))\in \bigoh(f(k))$ algorithm for \partition. 

The algorithm described in this section gives us the following.

\begin{lemma} \label{le:compute-mu-nu}
Assume that the feasible set $\mathcal F_i$ of each subgraph $G_i$ of $G$ is known. Then the feasible set $\mathcal F$ of $G$ can be computed in $\bigoh(f(k))$ time. 
\end{lemma}

With Lemma~\ref{le:compute-mu-nu} at hand, the algorithm to test whether $G$ admits an upward plane embedding with a given edge $e$ on the outer face, where $e$ is incident to $f_{\mathcal O}$, follows a natural recursive strategy described below.

\begin{figure}[htb]
	\centering
	\includegraphics[scale=0.7]{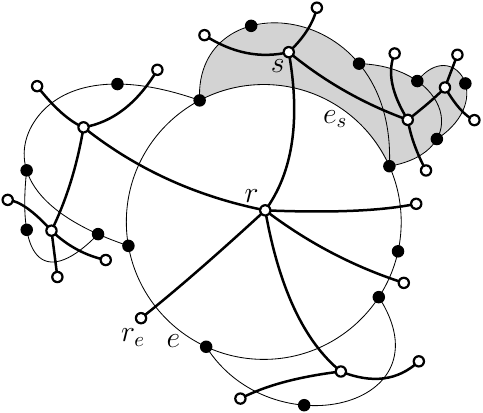}
	\caption{The extended dual tree $\mathcal T$ of an outerplane embedding $\mathcal O$ of $G$; the tree $\mathcal T$ is represented by white disks and thick lines. The subgraph $G_s$ of $G$ corresponding to the subtree $\mathcal T_s$ of $\mathcal T$ rooted at a node $s$ has its internal faces colored gray.}
	\label{fig:extended-dual}
\end{figure}

First, we compute the outerplane embedding $\mathcal O$ of $G$ in $\bigoh(n)$ time~\cite{cnao-lta-85,ht-ept-74,m-laarogmog-79,w-rolt-87}. Then we construct the \emph{extended dual tree} $\mathcal T$ of $\mathcal O$; see Fig.~\ref{fig:extended-dual}. This has a vertex $s$ for each cycle $\mathcal C_s$ delimiting an internal face of $\mathcal O$ and has a vertex for each edge incident to $f_{\mathcal O}$. Two vertices $s$ and $t$ of $\mathcal T$ corresponding to cycles $\mathcal C_s$ and $\mathcal C_t$ are adjacent in $\mathcal T$ if $\mathcal C_s$ and $\mathcal C_t$ share an edge, and a vertex of $\mathcal T$ corresponding to an edge $x$ is only adjacent to the vertex of $\mathcal T$ corresponding to the unique cycle delimiting an internal face of $\mathcal O$ the edge $x$ belongs to. We root $\mathcal T$ at the leaf $r_e$ corresponding to $e$. Then the subtree $\mathcal T_s$ of $\mathcal T$ rooted at an internal node $s$ of $\mathcal T$ corresponds to a subgraph $G_s$ of $G$, which is the union of the cycles corresponding to the vertices in $\mathcal T_s$. 

We now perform a bottom-up visit of $\mathcal T$ which ends after the child $r$ of $r_e$ is visited. When processing an internal node $s$ of $\mathcal T$ with parent $p$, we apply Lemma~\ref{le:compute-mu-nu} in order to compute the feasible set of $G_s$. We denote this feasible set by $\mathcal F_{s}$ and remark that the $(\mu,\nu)$-embeddings of $G_s$ are such that the edge $e_s$ shared by $\mathcal C_s$ and $\mathcal C_p$ -- this is the edge $e$ if $s=r$ -- is incident to the outer face. By Lemma~\ref{le:compute-mu-nu}, it takes $\bigoh(f(k_s))$ time to process $s$, where $k_s$ is the number of vertices  of $\mathcal C_s$. This sums up to $\bigoh(f(n))$ time over the entire bottom-up visit. If some computed feasible set $\mathcal F_{s}$ is empty, then by Property~$\mathcal P_3$ of Lemma~\ref{le:structure} the feasible set of every node in the path from $s$ to $r$ in $\mathcal T$ is empty, hence $G$ admits no upward plane embedding in which $e$ is incident to the outer face. Conversely, if a non-empty feasible set is computed for the child $r$ of the root $r_e$, then it is concluded that $G$ admits an upward plane embedding in which $e$ is incident to the outer face. The described algorithm can also be suitably modified so that, in case $G$ admits an upward plane embedding in which $e$ is incident to the outer face, it constructs such an embedding. In order to do that, whenever we add a pair $(\mu,\nu)$ to a feasible set $\mathcal F_{s}$ during the bottom-up visit of $\mathcal T$, we also associate to the pair the $\bigoh(k_s)$ values $\rho_0, \dots,\rho_k,\mu_1,\dots,\mu_k,\nu_1,\dots,\nu_k$ and the in-out assignment $\mathcal A$ that satisfy Properties~$\mathcal P_1$--$\mathcal P_7$ of Lemma~\ref{le:structure}. This information can be then exploited in a top-down traversal of $\mathcal T$ in order to compute a $(\mu,\nu)$-embedding of $G$. This concludes the proof of Theorem~\ref{th:upward-to-partition}.

\section{Conclusions} \label{se:conclusions}

In this paper, we have established a connection between the problem of testing the upward planarity of an $n$-vertex biconnected outerplanar DAG and a famous complexity theory problem, called \partition, which asks whether a set of positive integers admits a partition into two sets so that the sum of the elements in each set is equal to $n/2$. Our main result is that, if the upward planarity of an $n$-vertex biconnected outerplanar DAG can be tested in $f(n)$ time, then \partition can also be solved in $\bigoh(f(n))$ time. The time complexity of \partition, when expressed solely as a function of $n$, is known to be between $\Omega(n)$ and $\bigoh(n\log^2 n)$. Any super-linear lower bound for the \partition problem would hence imply a super-linear lower bound for testing the upward planarity of an $n$-vertex biconnected outerplanar DAG; this would be, to the best of our knowledge, the first super-linear lower bound for a graph drawing problem not involving real numbers. Conversely, a linear-time algorithm  for testing the upward planarity of an $n$-vertex biconnected outerplanar DAG would imply a linear-time algorithm for \partition, which would be a result improving upon the outcomes of decades of research.

We have also shown a result in the other direction. Namely, suppose that \partition can be solved in $f(n)$ time. Let $G$ be an $n$-vertex biconnected outerplanar DAG and $e$ be a prescribed edge incident to the outer face of the outerplane embedding of $G$. Then it can be tested in $\bigoh(f(n))$ time whether $G$ with admits an upward planar drawing in which $e$ is incident to the outer face. A result of Koiliaris and Xu~\cite{DBLP:journals/talg/KoiliarisX19} implies that such a function $f(n)$ exists with $f(n)\in \bigoh(n\log^2 n)$. On the one hand, the constraint that $e$ is incident to the outer face of the outerplane embedding $\mathcal O$ of $G$ is not a loss of generality per se, as we can prove that if a biconnected outerplanar DAG admits an upward plane embedding, it also admits an upward plane embedding $\mathcal E$ in which an edge incident to the outer face of $\mathcal O$ is incident to the outer face of $\mathcal E$. On the other hand, it would be nice to remove entirely the constraint of having a prescribed edge incident to the outer face and, possibly, to remove the assumption of biconnectivity. It is often the case that an algorithm to test some planarity variant (e.g., upward planarity or rectilinear planarity) with the constraint of having a prescribed edge that is required to be incident to the outer face can be used as a subroutine by an algorithm to solve the same problem, within the same time bound and without the constraint on the prescribed edge. A general methodology for doing that has been recently introduced~\cite{dlop-ood-20} and already used successfully for problems related to ours~\cite{DBLP:conf/gd/ChaplickGFGRS22,DBLP:journals/comgeo/Frati22}. However, the techniques introduced in~\cite{dlop-ood-20} do not seem to apply straight-forwardly to extend our algorithm so to deal with general biconnected and connected outerplanar DAGs.

We conclude by remarking that our paper shows a substantial difference between rectilinear planarity testing and upward planarity testing, two problems that are often tied together. Namely, while it can be tested in linear time whether an outerplanar graph admits a rectilinear planar drawing~\cite{DBLP:journals/comgeo/Frati22}, the results of Section~\ref{se:from-partition-to-upward} provide evidence that finding a linear-time algorithm for testing the upward planarity of an outerplanar DAG might be an elusive goal. Interestingly, both the linear-time algorithm for testing rectilinear planarity~\cite{DBLP:journals/comgeo/Frati22} and the algorithm we proposed in Section~\ref{se:from-upward-to-partition} exploit a bottom-up approach on the extended dual tree, in which the computation at each node of the tree mainly consists of deciding whether each graph component that is attached to the cycle corresponding to the node should be embedded inside or outside such a cycle. However, this choice has a different nature in the two problems. Indeed, in the rectilinear planarity testing problem, choices performed for distinct components are independent from one another. Conversely, in upward planarity testing, the constraints imposed by Theorem~\ref{th:upward-conditions} bound together some components in the inside/outside decision. This bond between components generates integers: Each integer corresponds to the number of times the cycle has to ``roll up'' in order to accommodate components that are bound together and that all lie inside (or all lie outside) the cycle. Theorem~\ref{th:upward-conditions} requires to balance the sum of the integers corresponding to components lying inside the cycle with the sum of the integers corresponding to components lying outside the cycle. This is where the connection to \partition arises, which makes upward planarity seemingly harder to test than rectilinear planarity.


		
	\bibliographystyle{splncs04}
	\bibliography{bibliography}
	

	\end{document}